\documentclass[letterpaper, 10 pt, conference]{ieeeconf}

\IEEEoverridecommandlockouts
\renewcommand{\baselinestretch}{0.976783}

\usepackage{float}
\usepackage{paralist}
\usepackage{amsfonts}
\usepackage{graphics}
\usepackage{epsfig}
\usepackage{mathptmx}
\usepackage{times}
\usepackage{amsmath}
\usepackage{amssymb}
\usepackage{physics}
\usepackage{mdwmath}
\usepackage{mdwtab}
\usepackage{color}
\usepackage[hidelinks]{hyperref}
\usepackage{algpseudocode}
\usepackage{algorithm}
\usepackage{caption}
\usepackage{subcaption}
\usepackage{bm}
\usepackage{tikz}
\usetikzlibrary{shapes, arrows.meta, positioning}
\usepackage{booktabs}
\usepackage{tabularx}
\usepackage{multirow}

\newtheorem{theorem}{Theorem}

\newtheorem{proposition}{Proposition}

\newtheorem{remark}{Remark}

\newtheorem{assumption}{Assumption}

\begin{document}

\title{\LARGE \bf
Certified Quantum Schrödinger Control via Hierarchical Tucker Models
}

\author{Nahid Binandeh Dehaghani, Rafal Wisniewski, A. Pedro Aguiar
\thanks{N. Dehaghani and R. Wisniewski are with Department of Electronic Systems, Aalborg University, Fredrik Bajers vej 7c, DK-9220 Aalborg, Denmark
        {\tt\small \{nahidbd,raf\}@es.aau.dk}}%
\thanks{A. Pedro Aguiar is with the Research Center for Systems and Technologies (SYSTEC), Electrical and Computer Engineering Department, FEUP - Faculty of Engineering, University of Porto, Rua Dr. Roberto Frias sn, i219, 4200-465 Porto, Portugal
        {\tt\small pedro.aguiar@fe.up.pt}}
}

\maketitle

\begin{abstract}
High-dimensional Schrödinger systems arising from tensor‑product discretizations suffer from exponential state growth, making direct controller synthesis and real‑time closed‑loop simulation computationally challenging. 
Hierarchical Tucker (HT) tensor representations offer scalable low‑rank surrogates, but the impact of fixed‑rank truncation on closed‑loop stability is not well understood. This paper develops a local robustness framework for sampled‑data feedback control implemented with fixed‑rank HT projections. By viewing each truncation as a bounded, rank‑dependent perturbation of the nominal closed loop, and assuming a local phase‑invariant contraction certificate together with trajectory‑level hierarchical spectral decay, we show that the HT‑projected dynamics are practically exponentially stable: trajectories converge to a dimension‑independent tube whose radius decreases with the prescribed rank. We further obtain an explicit logarithmic rank–accuracy relation and establish conditions under which controllers designed on the HT‑truncated surrogate model retain practical exponential tracking guarantees when deployed on the full system, together with an explicit bound quantifying the resulting surrogate‑to‑plant mismatch. A compact lattice example demonstrates the applicability of the framework.
\end{abstract}

\section{Introduction}
High-dimensional controlled Schr\"odinger equations arise in many areas of science and engineering, including quantum control~\cite{dalessandro2007intro,altafini2012quantum}, distributed wave phenomena~\cite{TeschlQM}, and quantum-inspired formulations of fluid dynamics such as the hydrodynamic Schr\"odinger equation (HSE)~\cite{meng2023HSE}. After discretization of a $p$-dimensional spatial domain $\Omega \subset \mathbb{R}^p$, the resulting state dimension typically grows exponentially with the number of degrees of freedom (e.g., $N=d^n$ for local dimension $d$). This curse of dimensionality renders direct controller synthesis and real-time closed-loop simulation computationally intractable, creating a major obstacle for feedback control of high-dimensional quantum systems.

Low-rank tensor representations provide an effective approach to mitigating this challenge. In particular, the hierarchical Tucker (HT) format enables scalable representations of high-dimensional tensors by exploiting hierarchical low-rank structure~\cite{Hackbusch2019TensorBook,grasedyck2010hierarchical}.
When the state remains close to a low-rank HT manifold, storage and propagation costs depend primarily on the hierarchical rank rather than on the ambient Hilbert-space dimension. 
Among low‑rank tensor formats, the HT representation is particularly well suited for high‑dimensional PDEs due to its balanced binary tree structure and dimension‑independent conditioning \cite{Hackbusch2019TensorBook}.
These properties make HT representations attractive for constructing tractable surrogate models of high-dimensional quantum systems. Evidence from many‑body physics further supports the compressibility that HT methods exploit: locality implies finite information‑propagation speed (Lieb–Robinson bounds), and low‑entanglement regimes are captured by area‑law behavior of entanglement entropy; together these mechanisms lead to rapid decay of Schmidt spectra across physically relevant bipartitions and hence low‑rank approximability over finite horizons \cite{lieb1972finite, EisertCramerPlenio2010}. 

Despite their success in numerical simulation, hierarchical tensor methods have been studied mainly in open-loop settings. In feedback control applications, the tensor representation must enforce a fixed rank at each time step, typically through hierarchical SVD truncation
\cite{grasedyck2010hierarchical} and related thresholding schemes \cite{BachmayrSchneider2017}. Such projections introduce approximation errors that perturb the nominal dynamics, raising important questions about stability and tracking performance when controllers are implemented on 
surrogate HT tensor models.
This paper develops a stability‑certified framework for feedback control of high-dimensional Schr\"odinger systems using fixed-rank hierarchical Tucker representations. 

\paragraph*{Contributions} The technical core of the paper lies in interpreting HT truncation as a structured disturbance and analyzing its interaction with contraction‑based stability.
The main contributions of this paper are:
(i) We show that enforcing a fixed HT rank via hierarchical SVD truncation introduces a structured perturbation whose norm decays exponentially with the hierarchical rank.
(ii) Under an incremental contraction assumption on the nominal closed loop, the projected dynamics remain practically exponentially attractive to a forward-invariant reference set, with an asymptotic error bound that decreases exponentially with the rank.
(iii) We establish conditions under which controllers synthesized on the HT‑truncated surrogate model retain practical exponential tracking guarantees when applied to the full system.
(iv) We derive a rank–accuracy design rule linking the hierarchical rank to a desired tracking tolerance and provide a projected closed-loop scheme based on operator splitting and fixed-rank HT truncation.


\section{Controlled Schr\"odinger Model}
\label{sec:model}
This section presents the continuum and semi-discrete controlled Schr\"odinger dynamics, together with the tensorized finite-dimensional representation used in the sequel. It also introduces the sampled-data discrete-time flow and the nominal closed-loop assumptions that underlie the later robustness analysis.
\subsection{Continuous and Semi-Discrete Schr\"odinger Dynamics}
\label{subsec:continuous_semidiscrete}

Consider controlled Schr\"odinger dynamics in bilinear form
\begin{equation}
  i\hbar\,\partial_t \psi(x,t)
  = \big(H_0 + u(t)H_1 \big)\psi(x,t),
  \quad \!\! \! \psi(\cdot,t)\in L^{2}(\Omega;\mathbb{C}),
  \label{eq:schrodinger}
\end{equation}
where $H_0$ is a densely defined self-adjoint operator generating the free dynamics and $H_1$ is a self-adjoint control Hamiltonian defined on the same invariant domain. The input $u:\mathbb{R}_{\ge0}\to\mathbb{R}$ is a measurable scalar control amplitude.
Throughout, $\Omega\subset\mathbb{R}^p$ is a bounded Lipschitz domain equipped with boundary conditions (e.g., Dirichlet or periodic) chosen so that $H_0$ is self-adjoint on a common invariant domain; see~\cite{TeschlQM,ReedSimonII}. This continuum model serves as the starting point for the sequel and will later be replaced by a semi-discrete finite-dimensional tensorized representation.



\begin{assumption}[Well-posedness]\label{ass:wellposed}
The operators $H_0$ and $H_1$ are densely defined and self-adjoint on a common invariant domain $\mathcal{D}\subset L^2(\Omega)$, and the control input satisfies $|u(t)| \le u_{\max}$ for all $t \ge 0$. For every admissible control, the Hamiltonian
$
H(t) = H_0 + u(t) H_1
$
is self-adjoint on $\mathcal{D}$.
\end{assumption}

\begin{remark}[Sufficient conditions]
Assumption~\ref{ass:wellposed} is standard in bilinear Schr\"odinger control. Typical sufficient conditions include the case where $H_1$ is bounded and the case where $H_1$ is $H_0$-bounded with relative bound strictly smaller than $1$, together with a uniformly bounded control input. In these cases, the self-adjointness of $H(t)=H_0+u(t)H_1$ on $\mathcal D$ follows from the Kato--Rellich theorem; see, e.g.,~\cite{ReedSimonII}. Under Assumption~\ref{ass:wellposed}, the evolution preserves the $L^2$ norm, i.e.,
$\|\psi(t)\|_{L^2}=\|\psi(0)\|_{L^2}$, $t\ge 0$.
\end{remark} 

\paragraph*{Spatial discretization and tensor structure}
For numerical implementation and low-rank tensor approximation, we discretize
$\Omega\subset\mathbb{R}^p$ on a tensor-product grid with $n_j$ degrees of freedom
along each spatial direction $j=1,\dots,p$. Let
$N = \prod_{j=1}^{p} n_j$
denote the total number of spatial degrees of freedom. To enable hierarchical
tensor representations, we further assume that the same discrete state space can
be tensorized as
$\mathcal{H}_N \cong \bigotimes_{i=1}^{n} \mathbb{C}^{d_i}$,
so that
$N=\prod_{i=1}^n d_i$.
Thus, the quantities $n_j$ describe the physical spatial discretization, whereas
$n$ and $d_i$ describe the chosen tensorization used for low-rank representation.
For simplicity, we restrict to the uniform case $d_i=d$, so that $N=d^{\,n}$,
where $n$ denotes the number of tensor factors in the chosen tensorization and
need not coincide with the physical dimension $p$. Such a factorization may
arise either from a tensor-product basis construction or from a structured
reshaping of the discrete state vector.

The semi-discrete wavefunction is therefore represented as $\Psi(t)\in\mathbb{C}^{N}$. Since the ambient dimension $N=d^n$ grows exponentially with the number of tensor factors $n$, direct simulation and feedback computation rapidly become impractical. This motivates low-rank tensor representations, in particular the hierarchical Tucker format used in this work. Throughout, dimension-independent refers to the absence of exponential dependence on the ambient dimension;
constants may still depend polynomially on structural quantities such as $n$, locality parameters, and the chosen tensor-network topology.

Applying a spatial discretization to~\eqref{eq:schrodinger} at dimension $N$ yields the semi-discrete controlled Schr\"odinger system
\begin{equation}
i\hbar\,\dot{\Psi}(t)
= \big(H_0^N + u(t)H_1^N\big)\Psi(t),
\qquad H_0^N,\,H_1^N \in \mathbb{C}^{N\times N},
\label{eq:discrete}
\end{equation}
where the superscript $N$ indicates dependence on the chosen discretization level, with $N=\prod_{j=1}^p n_j$ the total number of spatial degrees of freedom. Thus, $\{H_0^N,H_1^N\}_N$ should be understood as a family of finite-dimensional approximations of the continuum operators $H_0$ and $H_1$. We assume that, for each discretization level $N$, the matrices $H_0^N$ and $H_1^N$ are Hermitian discretizations of $H_0$ and $H_1$. This is satisfied by standard symmetric discretization schemes with compatible boundary conditions. Under this assumption, the semi-discrete evolution, and hence its exact time-$\Delta t$ propagator, is unitary; see, e.g.,~\cite{KohlmannSchrodingerSpectra,SolovejHilbertSpace}.

\begin{assumption}[Finite-range locality]
\label{ass:locality}
For each discretization level $N$, the matrices $H_0^N$ and $H_1^N$ admit decompositions
$H_0^N = \sum_{\ell=1}^{L_0(N)} h_\ell,$
$H_1^N = \sum_{\ell=1}^{L_1(N)} \tilde h_\ell,$
where each local term acts nontrivially only on a uniformly bounded subset of tensor factors in the chosen tensorization of $\mathbb C^N$. We assume that the associated interaction graph has degree bounded independently of $N$, that
$\|h_\ell\|,\;\|\tilde h_\ell\|\le C_h,$
with $C_h$ independent of $N$, and that the local tensor dimension $d$ remains fixed as the number of tensor factors $n$ increases.
\end{assumption}

Under Assumption~\ref{ass:locality}, the controlled Hamiltonian
$H^N(t)=H_0^N+u(t)H_1^N$ admits a decomposition into local terms, each acting
nontrivially only on a uniformly bounded subset of tensor factors. In this
sense, the interaction range remains uniformly bounded for all admissible
controls satisfying $|u(t)|\le u_{\max}$. In standard lattice settings, such locality is compatible with Lieb--Robinson-type bounds, which provide a finite effective propagation velocity for information and correlations; see, e.g.,~\cite{lieb1972finite}. 


\subsection{Discrete-Time Map and Nominal Assumptions}
\label{subsec:discrete_time_nominal}

For a time step $\Delta t>0$, and assuming the control is held constant over each interval $[k\Delta t,(k+1)\Delta t)$ with value $u_k$, define the discrete-time flow
\begin{equation}
\Psi_{k+1}
= \mathcal{F}_{\Delta t}(\Psi_k,u_k)
:= \exp\!\Big(
-\tfrac{i\Delta t}{\hbar}(H_0^N+u_kH_1^N)
\Big)\Psi_k .
\label{eq:discrete_flow}
\end{equation}
Whenever $|u_k|\le u_{\max}$, the generator $H_0^N+u_kH_1^N$ is Hermitian, and therefore $\mathcal{F}_{\Delta t}$ is unitary and norm preserving. In practice, the exponential may be approximated by operator-splitting or other structure-preserving time-stepping schemes, such as Lie--Trotter or Strang splittings~\cite{Trotter1959}. 

\paragraph*{Nominal closed-loop contraction}
We consider state-feedback laws of sampled-data form
$
u_k=\kappa(\Psi_k),
$
where the control is held constant on each interval
$[k\Delta t,(k+1)\Delta t)$.
The associated nominal closed-loop dynamics are
$
\Psi_{k+1}=\mathcal F_{\Delta t}(\Psi_k,\kappa(\Psi_k)).
$
Since Schr\"odinger dynamics are unitary in the ambient Euclidean norm, contraction is not a generic property of the uncontrolled flow. Instead, the analysis below assumes 
the feedback-designed nominal closed loop admits a local contraction certificate in a metric adapted to the control objective, relative to a forward-invariant reference set.


\begin{assumption}(Nominal contraction toward the reference set)
\label{ass:incremental_contraction}
Let $\mathcal{R}\subset\mathbb{C}^N$ be a closed forward-invariant
reference set for the nominal closed loop
$
\Psi_{k+1}=\mathcal{F}_{\Delta t}(\Psi_k,\kappa(\Psi_k)).
$
Assume there exist a metric
$\operatorname{dist}_\star$ on the local operating region and a constant
$\rho\in(0,1)$ such that, for all $\Psi$ in a neighborhood of $\mathcal{R}$,
$
\operatorname{dist}_\star\!\big(
\mathcal{F}_{\Delta t}(\Psi,\kappa(\Psi)),\mathcal{R}
\big)
\le
\rho\,\operatorname{dist}_\star(\Psi,\mathcal{R}).
$
\end{assumption}

Typical choices of $\operatorname{dist}_\star$ are phase-invariant distances on the quotient of the state space by the global-phase equivalence relation $\psi\sim e^{i\theta}\psi$. Examples include
$
\operatorname{dist}_\star(\psi,\phi)
:=
\min_{\theta\in\mathbb R}\|\psi-e^{i\theta}\phi\|_2,
$
and projector-based distances such as
$
\operatorname{dist}_\star(\psi,\phi)
:=
\|\psi\psi^\ast-\phi\phi^\ast\|_F.
$
Thus, $\operatorname{dist}_\star$ should be understood as a metric on the associated quotient space (equivalently, on pure states modulo global phase), rather than as a norm on the original Hilbert space. In the local region relevant for the analysis, we assume that $\operatorname{dist}_\star$ is locally equivalent to the Euclidean distance modulo global phase, up to fixed constants.

\begin{remark}
Assumption~\ref{ass:incremental_contraction} is a local hypothesis on the
nominal feedback-closed loop, not a general property of bilinear
Schr\"odinger dynamics. Such local contraction conditions arise in feedback
stabilization and tracking analyses for bilinear Schr\"odinger systems; see,
e.g.,~\cite{boscain2014multi}. In the present
work, this assumption serves as the nominal stability certificate whose robustness
under hierarchical tensor truncation is analyzed.
\end{remark} 

\section{Low-Rank Hierarchical Tensor Representation and Fixed-Rank Truncation}
\label{sec:tensor_manifold}
This section introduces the hierarchical Tucker representation used for the semi-discrete quantum state and then defines the fixed-rank truncation mechanism employed in the HT-truncated surrogate dynamics.
We first describe the geometric and algebraic structure of the HT representation, including dimension trees, nodewise matricizations, singular values, and hierarchical ranks. We then introduce the practical HSVD truncation operator, the associated residual bounds, and the spectral decay assumptions used later to control truncation errors along the closed-loop trajectory.

\subsection{Hierarchical Tucker Representation and Ranks}
\label{subsec:HT_ranks}

We work in the semi-discrete state space $\mathbb{C}^{N}$ with $N=d^{\,n}$ and identify each state vector $\Psi\in\mathbb{C}^{N}$ with an order-$n$ tensor in
$\bigotimes_{i=1}^{n}\mathbb{C}^{d}$ via the tensorization described in
Section~\ref{sec:model}. For exact quantum states one has $\|\Psi\|=1$, where
$\|\cdot\|$ denotes the Euclidean norm on $\mathbb{C}^{N}$ (equivalently, the
Frobenius norm after tensor reshaping).

\paragraph*{Hierarchical Tucker (HT) model and ranks}

Let $\mathcal T$ be a fixed dimension tree whose leaves correspond to the
tensor modes $\{1,\ldots,n\}$. For each internal node $t\in\mathcal T_{\mathrm{int}}$
with complementary index set $\bar t$, we denote by
$X^{(t)}(\Psi)$ the hierarchical matricization of $\Psi$
associated with the split $(t,\bar t)$.
Concretely, if $t=\{\mu\}$ and
$\bar t=\{1,\ldots,\mu-1,\mu+1,\ldots,n\}$,
then the matricization maps the tensor entries
$\Psi_{(t_1,\ldots,t_n)}$ into a matrix whose row index corresponds to
$t_\mu$ and whose column index corresponds to the multi-index
$(t_1,\ldots,t_{\mu-1},t_{\mu+1},\ldots,t_n)$, i.e.,
$X^{(t)}(\Psi)(t_\mu,\bar t)
=
\Psi_{(t_1,\ldots,t_\mu,\ldots,t_n)}.$ The singular values of $X^{(t)}(\Psi)$ are denoted by
$\{\sigma^{(t)}_\alpha\}_{\alpha\ge1}$.
The HT rank at
node $t$ is defined by
$
r_t(\Psi) := \operatorname{rank}\!\big(X^{(t)}(\Psi)\big).
$
A rank budget is specified by a tuple
$
\mathbf r := (r_t)_{t\in\mathcal T_{\mathrm{int}}},
$
and, when convenient, we write
$
r := \max_{t\in\mathcal T_{\mathrm{int}}} r_t
$
for the maximal prescribed nodewise rank. 
For this prescribed rank budget, the corresponding HT approximation class is
$
\mathcal H_{\mathbf r}
:=
\Big\{
\Psi \in \bigotimes_{i=1}^{n}\mathbb C^d
:\;
\operatorname{rank}\!\big(X^{(t)}(\Psi)\big)\le r_t
\quad \text{for all } t\in \mathcal T_{\mathrm{int}}
\Big\}.
$
The set $\mathcal H_{\mathbf r}$ is nonconvex, but it provides an efficient representation of multidimensional low-rank structure through a hierarchy of transfer tensors (HT cores). For detailed background on HT decompositions and parametrizations, see~\cite{grasedyck2010hierarchical,Hackbusch2019TensorBook}.

\subsection{Fixed-Rank Truncation and Spectral Decay Assumptions}
\label{subsec:HT_truncation_decay}

Given $\Psi$, one may consider the best approximation problem $\min_{\Phi\in\mathcal H_{\mathbf r}}\|\Psi-\Phi\|.$ Computing a globally optimal low-rank tensor approximation is generally difficult; for several tensor formats, such problems are NP-hard~\cite{HillarLim2013}. In the hierarchical Tucker format, rank reduction is performed in practice viahierarchical singular value decomposition (HSVD) truncation~\cite{grasedyck2010hierarchical,Hackbusch2019TensorBook}. Starting from a full or HT representation of $\Psi$, one orthogonalizes along the dimension tree and truncates the nodewise singular value decompositions to the prescribed ranks, yielding a practical rank-truncation operator $\widetilde{\Pi}_{\mathbf r}(\Psi)$.
Define the truncation residual by $\varepsilon_{\mathbf r}(\Psi) := \|\Psi-\widetilde{\Pi}_{\mathbf r}(\Psi)\|.$
Although HSVD truncation does not in general produce the globally optimal approximation in $\mathcal H_{\mathbf r}$, it satisfies a 
bound of the form
\begin{equation}\varepsilon_{\mathbf r}(\Psi)^2\;\le\;\sum_{t\in\mathcal T_{\mathrm{int}}}\;\sum_{\alpha>r_t}(\sigma^{(t)}_\alpha)^2 \label{eq:HT_global_tail_bound}\end{equation}

\paragraph*{Hierarchical spectral decay and compressibility}
Low-rank approximability in the HT format is governed by the decay of the
nodewise singular values. Such decay controls the truncation residual
in~\eqref{eq:HT_global_tail_bound} and therefore determines how accurately a
state tensor can be represented under a prescribed hierarchical rank budget.
We introduce two related assumptions: a static compressibility condition for an
individual tensor, and a trajectory-level condition ensuring that the same
compressibility persists along the closed-loop evolution. Such compressibility assumptions are consistent with the behavior of many-body
quantum systems exhibiting limited entanglement growth under local Hamiltonian
evolution. In particular, the effectiveness of tensor-network methods for
quantum lattice models is closely tied to the rapid decay of Schmidt spectra
across physically relevant bipartitions; see, e.g.,~\cite{EisertCramerPlenio2010}.

\begin{assumption}[Exponential hierarchical spectral decay]
\label{ass:HT_decay}
There exist constants $C_1,c>0$ such that for every internal node
$t\in\mathcal T_{\mathrm{int}}$ and all $\alpha\ge 1$,
$
\sigma^{(t)}_{\alpha}(\Psi) \le C_1 e^{-c\alpha}.
$
\end{assumption}


\begin{proposition}[Static HT truncation bound]\label{prop:HT_static_truncation} Under Assumption~\ref{ass:HT_decay}, in the uniform-rank case $r_t\equiv r$, there exist constants $C_2,c'>0$ such that $\varepsilon_{\mathbf r}(\Psi)\le C_2 e^{-c'r},$ where the constants may depend on the dimension tree $\mathcal T$ and the decay constants in Assumption~\ref{ass:HT_decay}.\end{proposition}

\begin{proof} This follows immediately from ~\eqref{eq:HT_global_tail_bound} and the geometric decay in Assumption~\ref{ass:HT_decay}.\end{proof}

We denote by $\Psi_k^{\,r}$ the rank-truncated state at step $k$, and set $u_k=\kappa(\Psi_k^{\,r})$.

\begin{assumption}(Uniform trajectory compressibility
)
\label{ass:HT_uniform}
There exist constants $C_1,c>0$, independent of $k$ and $t$, such that for every
time step, the post-step pre-truncation state
$
Z_k := \mathcal F_{\Delta t}(\Psi_k^{\,r},u_k)
$
satisfies
$
\sigma^{(t)}_{\alpha}(Z_k) \le C_1 e^{-c\alpha}
$
for every internal node $t\in\mathcal T_{\mathrm{int}}$ and all $\alpha\ge 1$. 
\end{assumption}


\begin{proposition}[Uniform trajectory-level truncation bound]
\label{cor:HT_uniform_truncation}
Under Assumption~\ref{ass:HT_uniform}, the pre-truncation states $Z_k$ satisfy
$
\|Z_k-\widetilde{\Pi}_{\mathbf r}(Z_k)\|
\le \bar\varepsilon_{\mathbf r}
$
uniformly in $k$, where in the uniform-rank case $r_t\equiv r$ one has
$
\bar\varepsilon_{\mathbf r}\le C_2 e^{-c' r}
$
for constants $C_2,c'>0$ independent of $k$.
\end{proposition}

It is also instructive to consider the case where the hierarchical
singular values decay algebraically rather than exponentially.
Suppose that
$\sigma^{(t)}_\alpha = O(\alpha^{-\beta})$
for some $\beta>\tfrac12$.
Then the nodewise tail estimate~\eqref{eq:HT_global_tail_bound}
implies
$
\sum_{\alpha>r}(\sigma^{(t)}_\alpha)^2
=
O\!\left(r^{-(2\beta-1)}\right)$. 
Consequently, in the uniform-rank setting the truncation error satisfies
$
\varepsilon_r
=
O\!\left(r^{-(\beta-\tfrac12)}\right)
$,
so that the HT approximation error decreases polynomially with the
prescribed rank; see, e.g.,~\cite{Hackbusch2019TensorBook,grasedyck2010hierarchical}.

\section{HT-Truncated Closed-Loop Dynamics}
\label{sec:projected_dynamics}

This section studies the 
surrogate closed loop obtained by applying fixed-rank HT truncation after each discrete-time Schr\"odinger update. We first isolate the truncated closed-loop map and quantify the size of the truncation perturbation. We then combine this perturbation bound with the nominal contraction hypothesis to derive practical stability and rank--accuracy estimates for the HT-truncated dynamics.

\subsection{Truncated Closed-Loop Map and Perturbation Bounds
}
\label{subsec:truncated_system_bounds}

We restrict the discrete-time Schr\"odinger evolution to the fixed-rank
HT approximation class $\mathcal H_{\mathbf r}$ introduced in
Section~\ref{sec:tensor_manifold}. Given the discrete flow
$\mathcal F_{\Delta t}$ from~\eqref{eq:discrete_flow} and the feedback law
$\kappa$, the truncated closed loop is defined by
\begin{equation}
\Psi_{k+1}^{\,r}
=
\widetilde{\Pi}_{\mathbf r}\!\big(
\mathcal F_{\Delta t}(\Psi_k^{\,r},\kappa(\Psi_k^{\,r}))
\big).
\label{eq:truncated_closed_loop}
\end{equation}
For analysis, define the post-step pre-truncation state
$
Z_k := \mathcal F_{\Delta t}(\Psi_k^{\,r},\kappa(\Psi_k^{\,r})),
$
and the truncation perturbation
$
e_k := \widetilde{\Pi}_{\mathbf r}(Z_k)-Z_k.
$
Then~\eqref{eq:truncated_closed_loop} can be written as
\begin{equation}
\Psi_{k+1}^{\,r}
=
\mathcal F_{\Delta t}(\Psi_k^{\,r},\kappa(\Psi_k^{\,r})) + e_k,
\label{eq:projected_dynamics}
\end{equation}
so the HT truncation enters the surrogate closed loop as an additive perturbation.
Under Assumption~\ref{ass:HT_uniform}, the magnitude of this perturbation is
uniformly controlled by the HSVD truncation error bound from
Section~\ref{sec:tensor_manifold}.

In practice, each 
surrogate step consists of three operations: evaluation of the feedback law $u_k=\kappa(\Psi_k^{\,r})$, propagation of the current HT state through the one-step map $\mathcal F_{\Delta t}$, and fixed-rank HSVD truncation of the resulting tensor. The analysis below concerns the resulting truncated map \eqref{eq:projected_dynamics} and does not depend on a particular implementation of the propagator, provided the assumptions of Section~\ref{sec:model} hold.

We next quantify the size of the truncation perturbation introduced at each step
of the truncated closed loop. 



\begin{proposition}[Uniform exponential truncation bound]
\label{prop:uniform_truncation_step}
Suppose Assumption~\ref{ass:HT_uniform} holds. Then there exist constants
$C_2,c'>0$ and a quantity $\bar\varepsilon_{\mathbf r}>0$, all independent of
the time index $k$, such that
$
\|e_k\|\le \bar\varepsilon_{\mathbf r},
$
$
\bar\varepsilon_{\mathbf r}\le C_2 e^{-c'r}.
$
\end{proposition}

\begin{proof}
Under Assumption~\ref{ass:HT_uniform}, Proposition~\ref{cor:HT_uniform_truncation}
shows that the HSVD truncation error of every pre-truncation state
$Z_k$ satisfies
$
\|Z_k-\widetilde{\Pi}_{\mathbf r}(Z_k)\|
\le
\bar\varepsilon_{\mathbf r},$
$
\bar\varepsilon_{\mathbf r}\le C_2 e^{-c'r},
$
uniformly in $k$. Since
$
e_k:=\widetilde{\Pi}_{\mathbf r}(Z_k)-Z_k,
$
it follows immediately that
$
\|e_k\|=\|Z_k-\widetilde{\Pi}_{\mathbf r}(Z_k)\|
\le
\bar\varepsilon_{\mathbf r},
$
and therefore
$
\|e_k\|\le C_2 e^{-c'r}.
$
\end{proof}

The constants in Proposition~\ref{prop:uniform_truncation_step} depend on the
HT dimension tree and the spectral decay parameters, but they are uniform in
time. In particular, for a fixed tree structure they do not exhibit exponential
dependence on the ambient Hilbert-space dimension $N=d^n$. Consequently, the
truncation perturbation can be made uniformly small along the surrogate closed-loop
trajectory by increasing the prescribed rank budget.

\subsection{Practical Stability and Rank--Accuracy Relations}
\label{subsec:practical_stability_rank_accuracy}

We now combine the uniform truncation bound with the nominal contraction
hypothesis from Assumption~\ref{ass:incremental_contraction}. The basic idea is
that the HT-truncated dynamics can be viewed as the nominal closed loop subject
to a bounded additive perturbation, whose size is controlled by the truncation
rank.

\begin{theorem}(Practical exponential stability of the HT-truncated closed loop)
\label{thm:practical_stability}
Assume Assumption~\ref{ass:incremental_contraction} holds for the nominal
closed loop with contraction factor $\rho\in(0,1)$ in the metric
$\operatorname{dist}_\star$. Let $\mathcal O\subset\mathbb C^N$ be a local
operating region containing the reference set $\mathcal R$, and suppose that
on $\mathcal O$ there exists a constant $M_\star\ge 1$ such that
$\operatorname{dist}_\star(x,y)\le M_\star \|x-y\|$
for all $x,y\in\mathcal O$. Suppose also that the truncation perturbation
satisfies $\|e_k\|\le \bar\varepsilon_{\mathbf r}$ for all $k\ge 0$.
Then, provided the HT-truncated trajectory remains in $\mathcal O$, for all
$k\ge 0$ one has
$\operatorname{dist}_\star(\Psi_k^{\,r},\mathcal R)
\le
\rho^k\,\operatorname{dist}_\star(\Psi_0^{\,r},\mathcal R)
+
\frac{M_\star\bar\varepsilon_{\mathbf r}}{1-\rho}$.
Moreover,
$\limsup_{k\to\infty}\operatorname{dist}_\star(\Psi_k^{\,r},\mathcal R)
\le
\frac{M_\star\bar\varepsilon_{\mathbf r}}{1-\rho}$.
\end{theorem}
\begin{proof}
The truncated iteration satisfies
$
\Psi_{k+1}^{\,r}
=
\mathcal{F}_{\Delta t}(\Psi_k^{\,r},\kappa(\Psi_k^{\,r})) + e_k,
\qquad
\|e_k\|\le \bar\varepsilon_{\mathbf r}.
$
By Assumption~\ref{ass:incremental_contraction},
$
\operatorname{dist}_\star\!\big(
\mathcal{F}_{\Delta t}(\Psi_k^{\,r},\kappa(\Psi_k^{\,r})),\mathcal R
\big)
\le
\rho\,\operatorname{dist}_\star(\Psi_k^{\,r},\mathcal R).
$
Since the HT-truncated trajectory remains in $\mathcal O$, the metric comparison
$\operatorname{dist}_\star(x,y)\le M_\star\|x-y\|$ is valid for the states under
consideration. Using the triangle inequality therefore yields
$
\operatorname{dist}_\star(\Psi_{k+1}^{\,r},\mathcal R)
\le
\rho\,\operatorname{dist}_\star(\Psi_k^{\,r},\mathcal R)
+
M_\star\bar\varepsilon_{\mathbf r}.
$
Iterating this bound gives
$
\operatorname{dist}_\star(\Psi_k^{\,r},\mathcal R)
\le
\rho^k \operatorname{dist}_\star(\Psi_0^{\,r},\mathcal R)
+
\frac{1-\rho^k}{1-\rho}\,M_\star\bar\varepsilon_{\mathbf r},
$
and the stated estimate follows. Taking the limit superior as $k\to\infty$
gives the asymptotic bound.
\end{proof}

\begin{remark}[Robust contraction interpretation]
The HT-truncated closed loop inherits the nominal contraction rate $\rho$ up to
a practical error tube whose radius is proportional to the uniform truncation
level $\bar\varepsilon_{\mathbf r}$. In this sense, the result is a local
input-to-state practical stability estimate with respect to the truncation
perturbation sequence $\{e_k\}$. In the uniform-rank case, the tube radius
decays exponentially with the prescribed rank budget under
Assumption~\ref{ass:HT_uniform}.
\end{remark}

The preceding result shows that, in the uniform-rank case $r_t\equiv r$, the
asymptotic tracking error of the HT-truncated closed loop decays exponentially
with $r$. This immediately yields a logarithmic rank--accuracy trade-off.

\begin{theorem}(Rank--accuracy relation under uniform exponential truncation)
\label{thm:rank_accuracy_uniform}
Suppose the nominal closed loop satisfies the contraction property of
Assumption~\ref{ass:incremental_contraction} with contraction factor
$0<\rho<1$, and let $\mathcal O\subset\mathbb C^N$ be the local operating
region from Theorem~\ref{thm:practical_stability}. Assume that on $\mathcal O$
there exists a constant $M_\star\ge 1$ such that
$\operatorname{dist}_\star(x,y)\le M_\star\|x-y\|$
for all $x,y\in\mathcal O$.
If the truncation perturbations satisfy the uniform exponential bound
$\|e_k\|\le \bar\varepsilon_{\mathbf r}\le C_2 e^{-c'r}$, $k\ge 0$,
as guaranteed by Proposition~\ref{prop:uniform_truncation_step},
then the HT-truncated closed loop satisfies
$\operatorname{dist}_\star(\Psi_k^{\,r},\mathcal R)
\le
\rho^k\,\operatorname{dist}_\star(\Psi_0^{\,r},\mathcal R)
+
\frac{M_\star C_2}{1-\rho}\,e^{-c'r}$,
$k\ge 0$.
In particular, to guarantee the asymptotic tracking tolerance
$\limsup_{k\to\infty}\operatorname{dist}_\star(\Psi_k^{\,r},\mathcal R)\le \eta$,
it is sufficient to choose the rank $r$ so that
$
r
\ge
\frac{1}{c'}\log\!(\frac{M_\star C_2}{(1-\rho)\eta}) =\mathcal O\!(\log(1/\eta)).$
\end{theorem}

\begin{proof}
Each truncation perturbation satisfies
$\|e_k\|\le \bar\varepsilon_{\mathbf r}\le C_2 e^{-c'r}$.
Substituting this bound into the estimate of Theorem~\ref{thm:practical_stability}
gives
$\operatorname{dist}_\star(\Psi_k^{\,r},\mathcal R)
\le
\rho^k\,\operatorname{dist}_\star(\Psi_0^{\,r},\mathcal R)
+
\frac{M_\star C_2}{1-\rho}\,e^{-c'r}$,
$k\ge 0$.
Taking the limit superior as $k\to\infty$ yields
$
\limsup_{k\to\infty}\operatorname{dist}_\star(\Psi_k^{\,r},\mathcal R)
\le
\frac{M_\star C_2}{1-\rho}\,e^{-c'r}.
$
Imposing
$\frac{M_\star C_2}{1-\rho}e^{-c'r}\le\eta$
and solving for $r$ yields the stated condition.
\end{proof}

Theorem~\ref{thm:rank_accuracy_uniform} shows that, under uniform exponential
compressibility along the trajectory, the required hierarchical rank grows only
logarithmically in the inverse tolerance $1/\eta$. In particular, no
horizon-dependent quantities appear: the uniform exponential truncation bound
ensures that the HT-truncated closed loop inherits practical exponential
stability from the nominal contracting closed loop with a steady-state error
that decays exponentially in~$r$.

\section{Controller Design on the Surrogate HT Model and Transfer to the Full System}


A central question in surrogate-model control is whether a controller defined on a surrogate model can still induce desirable closed-loop behavior on the underlying high-dimensional system. In the present setting, the feedback law is evaluated on the surrogate HT state $\Psi_k^{\,r}\in\mathcal H_{\mathbf r}$ rather than on the full state $\Psi_k\in\mathbb C^N$. The issue addressed in this section is therefore whether closed-loop guarantees established for the surrogate HT model can be transferred to the true full-order Schr\"odinger system.
This question is naturally related to abstraction-based control, in which an abstract model is linked to the concrete system through an error metric controlling the discrepancy between their evolutions; see, for example,~\cite{polaGirardTabuada,girardPappasSurvey}. In our setting, the HT-truncated dynamics provide the surrogate model, while the truncation bounds from Section~\ref{sec:projected_dynamics} quantify one source of mismatch between surrogate and full-order evolution. Because the controller is driven by $\Psi_k^{\,r}$ rather than the true full state $\Psi_k$, the transfer of closed-loop guarantees also requires control of the state-to-controller mismatch. The purpose of this section is to make these requirements explicit.

\paragraph*{Surrogate HT Model and Controller Design}

We consider controller synthesis directly on the surrogate HT approximation class
$\mathcal H_{\mathbf r}$. Let $\kappa_r:\mathcal H_{\mathbf r}\to\mathbb R$
denote a feedback law defined on the surrogate state. The surrogate HT
closed-loop map is
$
\mathcal T_{\mathbf r}(\Psi^{\,r})
:=
\widetilde{\Pi}_{\mathbf r}\!\big(
\mathcal F_{\Delta t}(\Psi^{\,r},\kappa_r(\Psi^{\,r}))
\big),$
$
\Psi_{k+1}^{\,r}=\mathcal T_{\mathbf r}(\Psi_k^{\,r}),$
$
\Psi_0^{\,r}\in\mathcal H_{\mathbf r}.
$
Thus, the surrogate controller is synthesized and analyzed entirely on the
HT-truncated surrogate dynamics.
Under the regularity assumptions, and away from rank-change singularities of
the HT representation, the map $\mathcal T_{\mathbf r}$ is locally well defined
on the operating region. The design objective is to choose $\kappa_r$ so that
the surrogate closed loop satisfies a local contraction certificate with respect
to the target reference set $\mathcal R$, namely,
$
\operatorname{dist}_\star\!\big(
\mathcal T_{\mathbf r}(\Psi^{\,r}),\mathcal R
\big)
\le
\rho\,\operatorname{dist}_\star(\Psi^{\,r},\mathcal R),
$
$
\rho\in(0,1),
$
for all surrogate states $\Psi^{\,r}$ in the operating region.

\paragraph*{Transfer of the surrogate controller to the Full Plant}

At run time, the surrogate controller computes the control input from the surrogate
state $\Psi_k^{\,r}\in\mathcal H_{\mathbf r}$ according to
$
u_k=\kappa_r(\Psi_k^{\,r}),
$
while the true full-order plant evolves according to
$
\Psi_{k+1}
=
\mathcal F_{\Delta t}(\Psi_k,u_k)
=
\mathcal F_{\Delta t}\big(\Psi_k,\kappa_r(\Psi_k^{\,r})\big).
$
In contrast, the surrogate model used for controller synthesis evolves as
$
\Psi_{k+1}^{\,r}
=
\mathcal T_{\mathbf r}(\Psi_k^{\,r})
=
\widetilde{\Pi}_{\mathbf r}\!\big(
\mathcal F_{\Delta t}(\Psi_k^{\,r},\kappa_r(\Psi_k^{\,r}))
\big).
$
Hence the plant--surrogate discrepancy is not due solely to HT truncation.
It also includes the mismatch between the true full state $\Psi_k$ and the
surrogate state $\Psi_k^{\,r}$ used to evaluate the controller. To capture both
effects simultaneously, we introduce the one-step transfer error
$
e_k^{\mathrm{full}}
:=
\mathcal F_{\Delta t}\big(\Psi_k,\kappa_r(\Psi_k^{\,r})\big)
-
\mathcal T_{\mathbf r}(\Psi_k^{\,r}).
$
A Surrogate–Plant transfer theorem therefore requires an explicit bound on
$\|e_k^{\mathrm{full}}\|$ (or on $\operatorname{dist}_\star$ applied to this
difference), rather than only the truncation residual of the surrogate model.
The following theorem formulates such a result.

\begin{theorem}(Surrogate–Plant transfer under one-step surrogate error)
\label{thm:surrogate_to_true}
Let
$
\mathcal T_{\mathbf r}(\Psi^{\,r})
=
\widetilde{\Pi}_{\mathbf r}\!\big(
\mathcal F_{\Delta t}(\Psi^{\,r},\kappa_r(\Psi^{\,r}))
\big)
$
be the surrogate HT closed‑loop map. Assume that the surrogate closed
loop admits a contraction certificate in the metric $\operatorname{dist}_\star$,
i.e., there exist a closed forward-invariant reference set $\mathcal R$ and a
contraction factor $\rho\in(0,1)$ such that
$
\operatorname{dist}_\star\!\big(\mathcal T_{\mathbf r}(\Psi^{\,r}),\mathcal R\big)
\le
\rho\,\operatorname{dist}_\star(\Psi^{\,r},\mathcal R)
$
for all $\Psi^{\,r}$ in 
a local operating region $\mathcal O$ containing the reference set $\mathcal R$.
Assume also that along the coupled
plant/surrogate evolution the one-step mismatch satisfies
$
\operatorname{dist}_\star\!\Big(
\mathcal F_{\Delta t}\big(\Psi_k,\kappa_r(\Psi_k^{\,r})\big),
\mathcal T_{\mathbf r}(\Psi_k^{\,r})
\Big)
\le
\delta_{\mathbf r},$
$k\ge 0,
$
for some rank-dependent quantity $\delta_{\mathbf r}\ge 0$, and that the surrogate
and full states satisfy the consistency bound
$
\operatorname{dist}_\star(\Psi_k^{\,r},\mathcal R)
\le
\operatorname{dist}_\star(\Psi_k,\mathcal R)+\delta_{\mathbf r}^{(0)},$ 
$k\ge 0,$ 
for some mismatch constant $\delta_{\mathbf r}^{(0)}\ge 0$.
Then, provided the plant and surrogate trajectories remain in the operating
region, 
the true plant trajectory generated by
$
\Psi_{k+1}=\mathcal F_{\Delta t}\big(\Psi_k,\kappa_r(\Psi_k^{\,r})\big)
$
satisfies
$
\operatorname{dist}_\star(\Psi_k,\mathcal R)
\le
\rho^k\,\operatorname{dist}_\star(\Psi_0,\mathcal R)
+
\frac{1-\rho^k}{1-\rho}\big(\rho\,\delta_{\mathbf r}^{(0)}+\delta_{\mathbf r}\big),
$
$k\ge 0.
$
In particular,
$
\limsup_{k\to\infty}\operatorname{dist}_\star(\Psi_k,\mathcal R)
\le
\frac{\rho\,\delta_{\mathbf r}^{(0)}+\delta_{\mathbf r}}{1-\rho}.
$
If both $\delta_{\mathbf r}$ and $\delta_{\mathbf r}^{(0)}$ are of order
$O(e^{-c'r})$, then the full plant inherits practical exponential attraction to
$\mathcal R$ with a tube radius of order $O(e^{-c'r})$.
\end{theorem}

\begin{proof}
Let
$
\Psi_{k+1}^{(\mathrm{surr})}=\mathcal T_{\mathbf r}(\Psi_k^{\,r}).
$
By the one-step mismatch assumption,
$
\operatorname{dist}_\star(\Psi_{k+1},\Psi_{k+1}^{(\mathrm{surr})})\le \delta_{\mathbf r}.$
By surrogate contraction,
$
\operatorname{dist}_\star(\Psi_{k+1}^{(\mathrm{surr})},\mathcal R)
\le
\rho\,\operatorname{dist}_\star(\Psi_k^{\,r},\mathcal R).
$
Hence, by the triangle inequality,
$
\operatorname{dist}_\star(\Psi_{k+1},\mathcal R)
\le
\rho\,\operatorname{dist}_\star(\Psi_k^{\,r},\mathcal R)+\delta_{\mathbf r}.
$
Using the consistency assumption gives
$
\operatorname{dist}_\star(\Psi_{k+1},\mathcal R)
\le
\rho\,\operatorname{dist}_\star(\Psi_k,\mathcal R)
+
\rho\,\delta_{\mathbf r}^{(0)}
+
\delta_{\mathbf r}.
$
Iterating this affine recursion yields
$
\operatorname{dist}_\star(\Psi_k,\mathcal R)
\le
\rho^k\,\operatorname{dist}_\star(\Psi_0,\mathcal R)
+
\frac{1-\rho^k}{1-\rho}\big(\rho\,\delta_{\mathbf r}^{(0)}+\delta_{\mathbf r}\big),
$
which proves the stated estimate. Taking the limit superior as $k\to\infty$
gives the asymptotic bound.
\end{proof}

\begin{remark}
The quantity $\delta_{\mathbf r}$ combines the truncation error of the HT‑truncated surrogate model with the controller-state mismatch induced by evaluating $\kappa_r$ on
$\Psi_k^{\,r}$ rather than on the true full state $\Psi_k$. Deriving explicit
bounds on $\delta_{\mathbf r}$ requires additional interface assumptions, such as
Lipschitz continuity of $\kappa_r$ and local consistency estimates between the
plant and surrogate states.
\end{remark}


\section{Numerical Validation}\label{sec:numerics}
We validate the framework on a \(4\times 4\) spin-\(\tfrac12\) lattice (Hilbert dimension \(2^{16}\)). The drift Hamiltonian is the nearest-neighbor Heisenberg model
\(
H_0 = J\!\sum_{\langle p,q\rangle}\big(\sigma_p^x\sigma_q^x+\sigma_p^y\sigma_q^y+\sigma_p^z\sigma_q^z\big),
\)
with coupling \(J=0.25\). The control acts on the two top-row sites,
\(
H_1=\tfrac12\big(\sigma_x^{(0,0)}+\sigma_x^{(0,1)}\big),
\)
and the sampling step is \(\Delta t=0.02\). The target is \(\phi_{\mathrm{tar}}=|1\rangle^{\otimes 16}\).
The sampled-data feedback is evaluated on the surrogate state as
$u_k=\gamma\,\mathrm{Im}\!\Big(\langle \phi_{\mathrm{tar}},H_1\Psi_k^{r}\rangle\;
\langle \Psi_k^{r},\phi_{\mathrm{tar}}\rangle\Big)$, 
$\gamma=3.0$, $|u_k|\le 3$,
and distances are measured with the phase-invariant metric \(\mathrm{dist}_\star\).
The discrete map \(\mathcal{F}_{\Delta t}\) is realized by a structure-preserving second-order Strang splitting built from the two-site interactions and the one-site control, consistent with the locality/unitarity setting in Section~II. The HT surrogate is updated as
\(
\Psi_{k+1}^{r}=\Pi_r\,\mathcal{F}_{\Delta t}\!\big(\Psi_k^{r},\kappa(\Psi_k^{r})\big),
\)
while the full plant evolves under the same surrogate-driven control,
\(
\Psi_{k+1}=\mathcal{F}_{\Delta t}\!\big(\Psi_k,\kappa(\Psi_k^{r})\big),
\)
matching the surrogate--plant transfer setup.
We simulate for \(K=220\) steps and consider ranks \(r\in\{2,4,8,12,16,24,32,64\}\).

\begin{figure}[t]
  \centering
    \includegraphics[width=\linewidth]{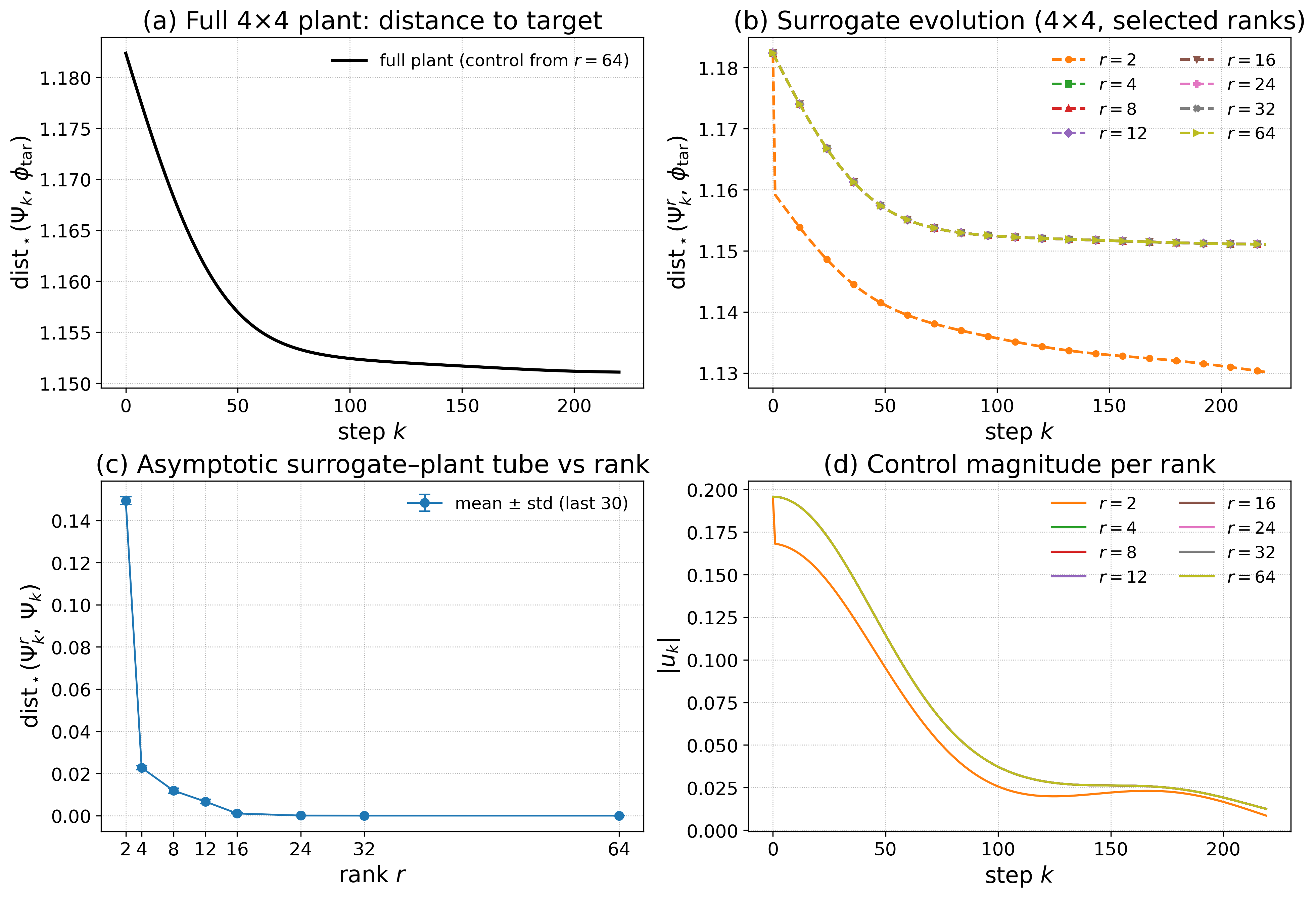} 
  \caption{\small Full-plant convergence and HT–truncated surrogate accuracy (4\(\times\)4 lattice).
  }
  \label{fig}
\end{figure}

Figure~\ref{fig}(a) shows the full plant converging toward the target behavior and settling into a small neighborhood, while the control magnitude decays smoothly from \(|u_0|\!\approx\!0.20\) to \(\approx\!0.02\) (Fig.~\ref{fig}(d)). %
In Fig.~\ref{fig}(b), the surrogate curves display the predicted rank dependence: smaller \(r\) yields larger steady distances, and beyond \(r\approx 8\) the curves nearly coincide, indicating that the truncation ceases to influence the evolution in this regime.
The surrogate--plant \emph{tube} in Fig.~\ref{fig}(c) is the tail average of the gap
\(\mathrm{dist}_\star(\Psi_k^{r},\Psi_k)\), i.e.,
\(\mathrm{Tube}(r)=\tfrac{1}{W}\sum_{k=K-W+1}^{K}\mathrm{dist}_\star(\Psi_k^{r},\Psi_k)\).
It drops sharply with rank (e.g., \(\approx 1.4\times10^{-1}\) at \(r=2\) to
\(\lesssim 5\times10^{-4}\) at \(r=64\)), showing that larger ranks reduce
truncation-induced discrepancy and shrink the asymptotic neighborhood around the
reference behavior.
Here \(\mathrm{dist}_\star\) is the phase-invariant Euclidean distance
\(\mathrm{dist}_\star(\psi,\phi)=\min_{\theta}\|\psi-e^{i\theta}\phi\|_2=\sqrt{2-2|\langle\phi,\psi\rangle|}\in[0,\sqrt{2}]\). Finally, the overlap of \(|u_k^{r}|\) sequences for \(r\ge 8\) (Fig.~\ref{fig}(d)) confirms that once the surrogate truncation is non-binding, the controller signals become rank-insensitive. 


\section{Conclusion}
We developed a stability-certified framework for feedback control of high-dimensional Schr\"odinger systems when the state is propagated with fixed-rank hierarchical Tucker projections. By viewing each truncation as a bounded, rank-dependent perturbation of a nominal contracting sampled-data closed loop, we showed that the projected dynamics remain locally stable and converge toward the reference behavior up to a rank-dependent neighborhood. The analysis also provides a clear rank--accuracy relation: the rank required to meet a desired asymptotic tolerance grows only logarithmically with the inverse of that tolerance.
We further established conditions under which controllers designed directly on the HT-truncated surrogate can be applied to the full system while preserving convergence to the same reference set, up to an explicitly quantified rank-dependent discrepancy.
A compact lattice example illustrates the nominal contraction, the rank–accuracy trend, and the compressibility observed along closed-loop trajectories. Future work includes extending the analysis to nonuniform trees and adaptive ranks, refining the surrogate-to-plant transfer bounds, and testing the approach on larger lattice systems and PDE discretizations.

\bibliographystyle{IEEEtran}
\bibliography{IEEEabrv,main}

\begin{thebibliography}{10}
\providecommand{\url}[1]{#1}
\csname url@samestyle\endcsname
\providecommand{\newblock}{\relax}
\providecommand{\bibinfo}[2]{#2}
\providecommand{\BIBentrySTDinterwordspacing}{\spaceskip=0pt\relax}
\providecommand{\BIBentryALTinterwordstretchfactor}{4}
\providecommand{\BIBentryALTinterwordspacing}{\spaceskip=\fontdimen2\font plus
\BIBentryALTinterwordstretchfactor\fontdimen3\font minus \fontdimen4\font\relax}
\providecommand{\BIBforeignlanguage}[2]{{%
\expandafter\ifx\csname l@#1\endcsname\relax
\typeout{** WARNING: IEEEtran.bst: No hyphenation pattern has been}%
\typeout{** loaded for the language `#1'. Using the pattern for}%
\typeout{** the default language instead.}%
\else
\language=\csname l@#1\endcsname
\fi
#2}}
\providecommand{\BIBdecl}{\relax}
\BIBdecl

\bibitem{dalessandro2007intro}
D.~D'Alessandro, \emph{Introduction to Quantum Control and Dynamics}.\hskip 1em plus 0.5em minus 0.4em\relax Chapman \& Hall/CRC, 2007.

\bibitem{altafini2012quantum}
C.~Altafini and F.~Ticozzi, ``Quantum mechanical bilinear control,'' \emph{IEEE Transactions on Automatic Control}, vol.~57, no.~8, pp. 1898--1917, 2012.

\bibitem{TeschlQM}
G.~Teschl, \emph{Mathematical Methods in Quantum Mechanics}.\hskip 1em plus 0.5em minus 0.4em\relax American Mathematical Society, 2014.

\bibitem{meng2023HSE}
Z.~Meng and Y.~Yang, ``Quantum computing of fluid dynamics using the hydrodynamic schr{\"o}dinger equation,'' \emph{Physical Review Research}, vol.~5, no. 033182, 2023.

\bibitem{Hackbusch2019TensorBook}
W.~Hackbusch, \emph{Tensor Spaces and Numerical Tensor Calculus}, 2nd~ed., ser. Springer Series in Computational Mathematics.\hskip 1em plus 0.5em minus 0.4em\relax Springer, 2019, vol.~56.

\bibitem{grasedyck2010hierarchical}
L.~Grasedyck, ``Hierarchical singular value decomposition of tensors,'' \emph{SIAM journal on matrix analysis and applications}, vol.~31, no.~4, pp. 2029--2054, 2010.

\bibitem{lieb1972finite}
E.~H. Lieb and D.~W. Robinson, ``The finite group velocity of quantum spin systems,'' \emph{Communications in Mathematical Physics}, vol.~28, no.~3, pp. 251--257, 1972.

\bibitem{EisertCramerPlenio2010}
J.~Eisert, M.~Cramer, and M.~B. Plenio, ``Area laws for the entanglement entropy,'' \emph{Reviews of Modern Physics}, vol.~82, no.~1, pp. 277--306, 2010.

\bibitem{BachmayrSchneider2017}
M.~Bachmayr and R.~Schneider, ``Iterative methods based on soft thresholding of hierarchical tensors,'' \emph{Foundations of Computational Mathematics}, vol.~17, pp. 1037--1083, 2017.

\bibitem{ReedSimonII}
M.~Reed and B.~Simon, \emph{Methods of Modern Mathematical Physics II: Fourier Analysis, Self-Adjointness}.\hskip 1em plus 0.5em minus 0.4em\relax Academic Press, 1975.

\bibitem{KohlmannSchrodingerSpectra}
\BIBentryALTinterwordspacing
M.~Kohlmann, ``Schr{\"o}dinger operators and their spectra,'' 2017, lecture notes, University of G{\"o}ttingen. [Online]. Available: \url{https://www.uni-math.gwdg.de/mkohlma/Documents/Schroedinger_operators_and_their_spectra.pdf}
\BIBentrySTDinterwordspacing

\bibitem{SolovejHilbertSpace}
\BIBentryALTinterwordspacing
J.~P. Solovej, ``Operators on hilbert space,'' 2008, lecture notes, University of Copenhagen. [Online]. Available: \url{https://web.math.ku.dk/~solovej/MATFYS/MatFys5.pdf}
\BIBentrySTDinterwordspacing

\bibitem{Trotter1959}
H.~F. Trotter, ``On the product of semigroups of operators,'' \emph{Proc. Amer. Math. Soc.}, vol.~10, pp. 545--551, 1959.

\bibitem{boscain2014multi}
U.~Boscain, M.~Caponigro, and M.~Sigalotti, ``Multi-input schr{\"o}dinger equation: controllability, tracking, and application to the quantum angular momentum,'' \emph{Journal of Differential Equations}, vol. 256, no.~11, pp. 3524--3551, 2014.

\bibitem{HillarLim2013}
C.~J. Hillar and L.-H. Lim, ``Most tensor problems are np-hard,'' \emph{SIAM Journal on Optimization}, vol.~23, no.~1, pp. 545--573, 2013.

\bibitem{polaGirardTabuada}
G.~Pola, A.~Girard, and P.~Tabuada, ``Approximately bisimilar symbolic models for nonlinear control systems,'' \emph{Automatica}, vol.~44, no.~10, pp. 2508--2516, 2008.

\bibitem{girardPappasSurvey}
A.~Girard and G.~J. Pappas, ``Approximate bisimulation: A bridge between computer science and control theory,'' \emph{European Journal of Control}, vol.~17, no. 5--6, pp. 568--578, 2011.

\end{thebibliography}
\end{document}